\documentclass[a4paper]{article}
\usepackage{a4wide}
\usepackage{amsmath}
\usepackage{amssymb}
\title{Common Permutation Problem}
\author{Mari\'an Dvorsk\'y\\\texttt{\normalsize marian.dvorsky@gmail.com}}
\date{\today}
\newtheorem{theorem}{Theorem}[section]
\newtheorem{lemma}{Lemma}[section]
\newtheorem{corollary}{Corollary}[section]
\newenvironment{example}[1][Example]{\begin{quotation}\noindent\textbf{Example.} }{\end{quotation}}
\newcommand{\qed}{\ifmmode\square\else{\unskip\nobreak\hfil\penalty 50\hskip 1em\null\nobreak\hfil$\square$\parfillskip=0pt\finalhyphendemerits=0\endgraf}\fi}
\newenvironment{proof}[1][Proof]{\begin{trivlist}\item[]\emph{#1. }\ignorespaces}{\qed\end{trivlist}}
\newcommand{\boundary}{\bullet}
\newcommand{\symb}[1]{\mathtt{#1}}
\newcommand{\prob}[1]{\texttt{#1}}
\newcommand{\symba}[1]{\symb{\triangle^{#1}}}
\newcommand{\symbb}[1]{\symb{\heartsuit^{#1}}}
\newcommand{\symbc}[1]{\symb{\Box^{#1}}}
\newcommand{\lp}{\langle}
\newcommand{\rp}{\rangle}
\newcommand{\varneg}[1]{\overline{#1}}
\newcommand{\symbvar}[2]{\symb{#1^{#2}}}
\newcommand{\symbvarneg}[2]{\symb{\varneg{#1}^{#2}}}
\begin{document}

\maketitle

\begin{abstract}
In this paper we show that the following problem is $NP$-complete: Given an alphabet $\Sigma$ and
two strings over $\Sigma$, the question is whether there exists a permutation
of $\Sigma$ which is a~subsequence of both of the given strings.
\end{abstract}

\section{Introduction}
\label{intro}

In computer science, efficient algorithms for various string problems are
studied. One of such problems is a well-known \prob{Longest Common Subsequence}
problem. For two given strings, the problem is to find
the longest string which is a subsequence of both the strings. A~survey of
efficient algorithms for this problem can be found in \cite{Survey}.

Let us consider a modification of the \prob{Longest Common Subsequence} problem. Instead of finding any
longest common subsequence, we restrict ourselves to subsequences in which
symbols do not repeat, i.e., every symbol occurs at most once. We call this problem
\prob{Longest Restricted Common Subsequence}.\footnote{A more general version of this
problem (with the same name) appeared in \cite{Andrejkova} together with its
efficient solution. Unfortunately, that solution is incorrect. Our result in
this paper indeed shows that an efficient (polynomial) solution for this
problem does not exist unless $P=NP$.}

\begin{example}
For strings ``$\symb{bcaba}$'' and ``$\symb{babcca}$'', the longest
common subsequence is ``$\symb{baba}$`` while the longest restricted common
subsequence is ``$\symb{bca}$''.
\end{example}

\prob{Longest Restricted Common Subsequence} is an optimization problem. In
this paper we consider its special case which is the following decision
problem: Suppose that the two strings are formed over an alphabet $\Sigma$. The
question is, do the two strings contain a restricted common subsequence of the
maximal possible length, i.e., a string that contains \textsl{every} symbol of $\Sigma$
exactly once? Such a string is a permutation of $\Sigma$. Therefore, we call
this problem the \prob{Common Permutation} problem.

\begin{quote}
\prob{Common Permutation}\\
\textsl{Instance}: An alphabet $\Sigma$ and two strings $a,b$ over $\Sigma$.\\
\textsl{Question}: Is there a permutation of $\Sigma$ which is a common
subsequence of $a$ and $b$?
\end{quote}

We will show that \prob{Common Permutation} is $NP$-complete. Moreover, we
will show that \prob{Common Permutation} is $NP$-complete even if the input
strings contain every symbol of $\Sigma$ at most twice.

\prob{Common Permutation} can be reduced to \prob{Longest Restricted
Common Subsequence} by asking whether the longest restricted common
subsequence of the two strings is equal to the size of the alphabet. Since
\prob{Common Permutation} will be shown to be $NP$-complete, it follows that \prob{Longest
Restricted Common Subsequence} is $NP$-hard.

In the next section we define the terms used in this paper. Section
\ref{alignments} introduces \textsl{alignments} as a way to visualize the \prob{Common Permutation}
problem. Finally, Section \ref{proof} presents the proof of
$NP$-completeness by reducing \prob{3SAT} to \prob{Common Permutation}.

\section{Preliminaries}

An \textsl{alphabet} is a finite set of \textsl{symbols}. A \textsl{string}
over an alphabet $\Sigma$ is a finite sequence $a=a_1a_2\dots a_N$ where
$N$ is a \textsl{length} of the string and $a_i\in\Sigma$ for all
$i\in\{1,\dots,N\}$. We say that $a_i$ is a
symbol on a \textsl{position} $i$ in the string $a$. For a given symbol $x\in
\Sigma$, \textsl{occurrences of $x$ in $a$} are all positions $i$ such that $a_i=x$.

A \textsl{subsequence} of a string $a=a_1a_2\dots a_N$ over $\Sigma$
is a string $b=a_{i_1}a_{i_2}\dots a_{i_n}$ where
$n \in \{0,1,\dots, N\}$ and $1\leq i_1 < i_2 < \cdots < i_n\leq N$. A
\textsl{common subsequence} of two strings $a$ and $b$ is a string which is a
subsequence of both $a$ and $b$.

A \textsl{permutation} of a finite set $A=\{x_1,\dots,x_n\}$ is a string
$x_{i_1}x_{i_2}\dots x_{i_n}$ (note that the length of the string is the same as
the number of elements in $A$) where $i_j\in \{1,\dots,n\}$ for
$j\in\{1,\dots,n\}$ and for all $k,l\in \{1,\dots,n\}$ if $k\neq l$ then
$i_k\neq i_l$.

The above definitions give a formal basis for the statement of the problem
from Section \ref{intro}.

\medskip

For the proof of $NP$-completeness in Section \ref{proof} we use the
reduction from \prob{3-Satisfiability} (\prob{3SAT} for short).  The following
definitions are from \cite{GareyJohnson}.

\subsection{3-Satisfiability}

Let $U=\{u_1,u_2,\dots,u_n\}$ be a set of Boolean variables. A \textsl{truth
assignment}
for $U$ is a function $t:U\to\{T,F\}$. If $t(u) = T$ we say that $u$ is true
under $t$; if $t(u) = F$ we say that $u$ is false. If $u$ is a variable in $U$,
then $u$ and $\varneg{u}$ are \textsl{literals} over $U$. The literal $u$ is true under $t$
if and only if the variable $u$ is true under $t$; the literal $\varneg{u}$ is
true if and only if the variable $u$ is false.

A \textsl{clause} over $U$ is a set of literals over $U$, for example
$\{u_1,\varneg{u_3},u_8\}$. It represents the disjunction of those literals and is
satisfied by a truth assignment if and only if at least one of its members is
true under that assignment. In other words, the clause is not satisfied if and
only all its literals are false. The clause above will be satisfied by $t$ unless
$t(u_1)=F$, $t(u_3) = T$, $t(u_8)=F$. A~collection $C$ of clauses over $U$ is
satisfiable if and only if there exists some truth assignment for $U$ that
simultaneously satisfies all the clauses in $C$. Such a truth assignment is
called a satisfying truth assignment for $C$.

\begin{quote}
\prob{3SAT}\\
\textsl{Instance}: A set $U$ of variables and a collection $C$ of clauses over $U$ with
exactly three literals per clause.\\
\textsl{Question}: Is there a satisfying truth assignment for $C$?
\end{quote}

\begin{theorem}
\prob{3SAT} is $NP$-complete.
\end{theorem}

See \cite{GareyJohnson} for the definition of $NP$-completeness and for the
proof of this theorem.

\section{Alignments}
\label{alignments}
In Section \ref{proof} we will use a notion of
\textsl{alignments}. Imagine the two input strings of \prob{Common Permutation}
written in two rows, one string per row. For every symbol of the alphabet
$\Sigma$ we want to find exactly one occurrence of that symbol in both
strings, such that we can ``align'' those occurrences. 

\begin{example}
For two strings ``$\symb{bcaba}$'' and ``$\symb{babcca}$'', one of the possible
alignments is depicted below (the aligned occurrences are bold)
\[\begin{array}{cccccc}
\mathbf{b}&  &\mathbf{c}&\symb{ab}&&\mathbf{a}\\
\mathbf{b}&\symb{ab}&\mathbf{c}&&\symb{c} &\mathbf{a}
\end{array}\]
\end{example}

Formally, let $a$ and $b$ be strings over an alphabet $\Sigma$. Let $n$ be the
number of symbols in $\Sigma$. An \textsl{alignment}
(denoted $A$) of $a$ and $b$ is a sequence of ordered pairs $A=\lp i_1, j_1 \rp,
\lp i_2, j_2 \rp, \dots, \lp i_n, j_n \rp$ such that for all $k$, the value of $i_k$ is a position in the
string $a$, $j_k$ is a position in the string $b$, and $a_{i_k} = b_{j_k}$.
Moreover, $i_1<\cdots<i_n$, $j_1<\cdots<j_n$, and 
$a_{i_1}a_{i_2}\dots a_{i_n} (= b_{j_1}b_{j_2}\dots b_{j_n})$ is a~permutation of $\Sigma$. 

For all $k$ we say that, in the alignment $A$, the position $i_k$ in the string $a$ \textsl{is
aligned with} the position $j_k$ in the string $b$. We also say that the symbol 
$a_{i_k} (=b_{j_k})$ is \textsl{aligned} at the position $i_k$ in $a$, and at
the position $j_k$ in $b$. Positions $i_k$ and $j_k$ are \textsl{aligned
occurrences} of $a_{i_k}$.

Notice that once a position $i$ (in $a$) is aligned with a position $j$ (in
$b$), positions
less than $i$ (in $a$) cannot be aligned with positions greater than $j$ (in $b$)
and vice versa. In
other words, the aligned occurrences of different symbols cannot ``cross''.

\begin{lemma}Let $a$ and $b$ be two strings over $\Sigma$. A
permutation of $\Sigma$ which is a common subsequence of $a$ and
$b$ exists if and only if there exists one or more alignments of $a$
and $b$.
\end{lemma}
\begin{proof}
An alignment corresponds to subsequences in $a$ and $b$ which comprise a
(common) permutation of $\Sigma$.
\end{proof}

According to this lemma, an alignment of two strings is an existence proof (of
a polynomial size with respect to lengths of the strings) for an instance of
\prob{Common Permutation}. Therefore, \prob{Common Permutation} is in $NP$. The
proof of $NP$-\textsl{completeness} follows in the next section.

\section{Reduction}
\label{proof}
In this section we will reduce \prob{3SAT} to \prob{Common Permutation}.

\begin{theorem}
\label{CPisNPC}
\prob{Common Permutation} is $NP$-complete.
\end{theorem}
\begin{proof}
Let $U$ be a finite set of variables and $C=\{c_1,c_2,\dots,c_n\}$
be a set of clauses over $U$. We have to construct an alphabet $\Sigma$ and two strings
$a,b$ over $\Sigma$ such that there exists a permutation of
$\Sigma$ which is a common subsequence of both $a$ and $b$ if and only if $C$ is
satisfiable.

The proof consists of two parts. The first part presents the construction of
$\Sigma$ and the strings. The second part proves that the construction
is correct in a sense that it satisfies the property described above.

\paragraph{Construction}
The alphabet $\Sigma$ consists of a pair of symbols $\symbvar{u}{i}$ and
$\symbvarneg{u}{i}$ for every variable $u\in U$ and every clause $c_i$ for
which either $u\in c_i$ or $\varneg{u} \in c_i$.
Additionally, $\Sigma$ contains a special ``boundary'' symbol $\boundary$.

The strings $a$ and $b$ have two parts: ``truth-setting'' part and ``satisfaction testing''
part. The parts are separated by the boundary symbol which ensures that
occurrences from one part cannot be aligned with occurrences from the other part of the
strings.

The ``truth-setting'' part consists of a concatenation of blocks, one for each
variable. Let $u$ be a variable from $U$ and let $\{i_1,\dots,i_m\}$ be the indexes of
clauses in which it appears. The strings contain the following block for
variable $u$:
\[\begin{array}{ccll}
a &=&\dots
\symbvar{u}{i_1}\symbvar{u}{i_2}\dots\symbvar{u}{i_m}
&\symbvarneg{u}{i_1}\symbvarneg{u}{i_2}\dots\symbvarneg{u}{i_m}
\dots\\
b &=&\dots
\symbvarneg{u}{i_1}\symbvarneg{u}{i_2}\dots\symbvarneg{u}{i_m}
&\symbvar{u}{i_1}\symbvar{u}{i_2}\dots\symbvar{u}{i_m}
\dots
\end{array}\]

This block is constructed in such a way that it is possible to simultaneously 
align all the symbols $\{\symbvar{u}{i_1},\symbvar{u}{i_2},\dots,\symbvar{u}{i_m}\}$
inside this block, or all the symbols
$\{\symbvarneg{u}{i_1},\symbvarneg{u}{i_2},\dots,\symbvarneg{u}{i_m}\}$.
It is, however, not possible to simultaneously align both $\symbvar{u}{i}$ and
$\symbvarneg{u}{j}$ for some $i$ and $j$ inside this block.

The ``satisfaction-testing'' part consists of a concatenation of blocks, one
for each clause. For a clause $c_i\in C$, let $x$, $y$, and $z$ be the
literals in the clause $c_i$, i.e., $c_i=\{x,y,z\}$. We use the
following notation:
\begin{itemize}
\item if $x = u$ for $u\in U$, then $\symbvar{x}{i} = \symbvar{u}{i}$ and  
$\symbvarneg{x}{i} = \symbvarneg{u}{i}$
\item if $x = \varneg{u}$ for $u\in U$, then 
$\symbvar{x}{i} = \symbvarneg{u}{i}$ and  
$\symbvarneg{x}{i} = \symbvar{u}{i}$
\end{itemize}

The strings contain the following block for the clause $c_i$:
\[\begin{array}{cclc}
a &=& \dots
\symbvar{x}{i}\symbvar{y}{i}\symbvar{z}{i}
&\symbvarneg{x}{i}\symbvarneg{y}{i}\symbvarneg{z}{i}
\dots\\
b &=& \dots
\symbvar{x}{i}\symbvar{y}{i}\symbvar{z}{i}
&\symbvarneg{y}{i}\symbvarneg{x}{i}\symbvarneg{z}{i}\symbvarneg{y}{i}
\dots
\end{array}\]

The block has two parts. The left part is the same for both strings. The
right part is constructed in such a way that the symbols
$\{\symbvarneg{x}{i},\symbvarneg{y}{i},\symbvarneg{z}{i}\}$ cannot be
simultaneously aligned in this block. Notice that these are the symbols
corresponding to the truth assignment for which the clause is false.

The alphabet $\Sigma$ contains $6n + 1$ symbols. The length of the string $a$ is
$6n + 1 + 6n=12n + 1$; the length of the string $b$ is $6n + 1 + 7n = 13n +
1$. Therefore, the size of the constructed \prob{Common Permutation} instance is polynomial
with respect to the original \prob{3SAT} instance. The construction can be
carried out in polynomial time.

\begin{example}
For a set of variables $\{w,x,y,z\}$ and
clauses $\{\{w,\varneg{x},y\},\{\varneg{z},x,\varneg{y}\}\}$ which represent
the logical function $$(w\lor \bar{x} \lor y)\land(\bar{z}\lor x\lor \bar{y})$$
we get the alphabet $$\Sigma=\{
\symbvar{w}{1},\symbvarneg{w}{1},
\symbvar{x}{1},\symbvarneg{x}{1},\symbvar{x}{2},\symbvarneg{x}{2},
\symbvar{y}{1},\symbvarneg{y}{1},\symbvar{y}{2},\symbvarneg{y}{2},
\symbvar{z}{2},\symbvarneg{z}{2}\}$$
and the following strings:
\[
\begin{array}{llllcll}
a = \symbvar{w}{1}\symbvarneg{w}{1}
&\symbvar{x}{1}\symbvar{x}{2}\symbvarneg{x}{1}\symbvarneg{x}{2}
&\symbvar{y}{1}\symbvar{y}{2}\symbvarneg{y}{1}\symbvarneg{y}{2}
&\symbvar{z}{2}\symbvarneg{z}{2}
&\boundary
&\symbvar{w}{1}\symbvarneg{x}{1}\symbvar{y}{1}\symbvarneg{w}{1}\symbvar{x}{1}\symbvarneg{y}{1}
&\symbvarneg{z}{2}\symbvar{x}{2}\symbvarneg{y}{2}\symbvar{z}{2}\symbvarneg{x}{2}\symbvar{y}{2}
\\
b = \underbrace{\symbvarneg{w}{1}\symbvar{w}{1}}_{\mbox{for }w}
&\underbrace{\symbvarneg{x}{1}\symbvarneg{x}{2}\symbvar{x}{1}\symbvar{x}{2}}_{\mbox{for }x}
&\underbrace{\symbvarneg{y}{1}\symbvarneg{y}{2}\symbvar{y}{1}\symbvar{y}{2}}_{\mbox{for }y}
&\underbrace{\symbvarneg{z}{2}\symbvar{z}{2}}_{\mbox{for }z}
&\boundary
&\underbrace{\symbvar{w}{1}\symbvarneg{x}{1}\symbvar{y}{1}\symbvar{x}{1}\symbvarneg{w}{1}\symbvarneg{y}{1}\symbvar{x}{1}}_{\mbox{for clause 1}}
&\underbrace{\symbvarneg{z}{2}\symbvar{x}{2}\symbvarneg{y}{2}\symbvarneg{x}{2}\symbvar{z}{2}\symbvar{y}{2}\symbvarneg{x}{2}}_{\mbox{for clause 2}}
\end{array}
\]
\end{example}

\paragraph{Correctness} Now we verify that the constructed strings $a$ and
$b$ contain a common permutation of $\Sigma$ if and only if $C$ is satisfiable.

Let $t:U\to\{T,F\}$ be any satisfying truth assignment for $C$. We will show that
there exists a permutation of $\Sigma$ which a common subsequence of both $a$
and $b$, i.e., that it is possible to align all symbols from
$\Sigma$ in the strings.

There is only one choice how to align the boundary symbol. For a variable $u\in U$,
if $t(u)=T$, we align $\symbvarneg{u}{i}$ symbols for all $i$ in the
``truth-assigning'' part and
$\symbvar{u}{i}$ in the ``satisfaction-testing'' part. If $t(u) = F$ we
conversely align $\symbvar{u}{i}$ symbols in the ``truth-assigning''
part and $\symbvarneg{u}{i}$ in the ``satisfaction-testing'' part.

As we noted during construction, the desired alignment in the ``truth-assigning'' part
is always possible to find. To show that we can align the remaining symbols in the
``satisfaction-testing'' part, we use that $t$ is satisfying $C$.

Notice that the symbols which we have to align
in the ``satisfaction-testing'' part correspond to the truth values of the variables. For example, if
we have to align symbol $\symbvarneg{u}{i}$ in this part, we know that $t(u)=F$.
For every clause $c_i\in C$, $c_i=\{x,y,z\}$, we align the
remaining symbols corresponding to clause $c_i$ in the block for $c_i$.
The symbols $\symbvar{x}{i}$, $\symbvar{y}{i}$, and $\symbvar{z}{i}$ can be
aligned in the first part of the block. It is easy to see that any pair of
symbols from $\{\symbvarneg{x}{i},\symbvarneg{y}{i},\symbvarneg{z}{i}\}$ can be aligned in
the second part. Therefore, for all seven possibilities how $c_i$ can be satisfied,
we can align the corresponding symbols in the block for the clause $c_i$.

For the proof in the opposite direction, suppose now that $a$ and $b$
have a common permutation of the symbols in
$\Sigma$. We will construct a satisfying truth assignment for $C$. For that
we look at ``truth-setting'' part of the strings. For a
variable $u\in U$,
\begin{itemize}
\item if the symbol $\symbvarneg{u}{i}$ for some $i$ is aligned in the
``truth-setting'' part of the strings, we set $t(u)=T$,
\item if the symbol $\symbvar{u}{i}$ for some $i$ is aligned in the
``truth-setting'' part of the strings, we set $t(u)=F$,
\item if none of the symbols $\{\symbvar{u}{i}, \symbvarneg{u}{i}\}$ are aligned in the
``truth-setting'' part, we set $t(u)$ arbitrarily, say $t(u)=T$.
\end{itemize}

Notice that, according to the construction of the
``truth-setting'' part, this is a valid definition of the assignment, i.e., it
cannot happen that we would want to assign $t(u)$ to both $T$ and $F$.

We now have to prove that $t$ is a satisfying truth assignment for $C$. For
any clause $c_i\in C$, let $x$, $y$, $z$ be its literals, so $c_i = \{x,y,z\}$.
We know that not all the symbols $\{\symbvarneg{x}{i}, \symbvarneg{y}{i},
\symbvarneg{z}{i}\}$ can be aligned in ``satisfaction-testing'' part of the
strings, so at least one them must be aligned in the ``truth-setting'' part.
Without loss of generality say it is $\symbvarneg{x}{i}$. Therefore, by the
definition of $t$, we know that the literal $x$ is true, and therefore $c_i$ is true.
\end{proof}

Our construction in the proof used every symbol at most twice in the string
$a$, but used some of the symbols three times in the string $b$. The following
corollary shows that we can use a slightly different construction which
uses every symbol at most twice in both the strings.

\begin{corollary}
\label{twosymbols}
\prob{Common Permutation} is $NP$-complete even if every symbol occurs at most
twice in the given strings.
\end{corollary}
\begin{proof}
We will use the same construction as in the proof of Theorem
\ref{CPisNPC}, except for the definition of $\Sigma$, and blocks for clauses.
For every clause $c_i\in C$ we will add three additional symbols $\symba{i}$, $\symbb{i}$, and $\symbc{i}$ to the
alphabet $\Sigma$. The strings contain the following block for the clause $c_i$:
\[\begin{array}{ccl}
a &=& \dots \symbvarneg{x}{i}\symbvar{x}{i}\symba{i}\symbb{i}\symbvar{y}{i}\symbvarneg{y}{i}\symba{i}\symbc{i}\symbb{i} \symbvarneg{z}{i} \symbc{i} \symbvar{z}{i}\dots\\
b &=& \dots \symbvar{x}{i} \symba{i}\symbvarneg{x}{i}\symbvarneg{y}{i}\symbb{i}\symbvar{y}{i}\symbc{i}\symba{i}\symbb{i}\symbc{i} \symbvar{z}{i}\symbvarneg{z}{i}\dots
\end{array}\]

One can verify that inside this block we can
align symbols corresponding to the satisfying assignment for $c_i$, but we cannot 
align simultaneously $\symbvarneg{x}{i}$, $\symbvarneg{y}{i}$, and
$\symbvarneg{z}{i}$.

The constructed strings contain every symbol exactly twice with the exception
of $\boundary$ which they contain once.
\end{proof}

\section{Acknowledgements}

The author would like to thank Martin Alter for suggesting Corollary
\ref{twosymbols}, J\'an Katreni\v c for discussions about the problem, and
authors of \cite{Crochemore} for bringing the problem (again) to author's attention.

\end{document}